\documentclass[letterpaper, 10 pt, conference]{ieeeconf}

\IEEEoverridecommandlockouts  

\usepackage{fix-cm}
\usepackage{etex}

\usepackage{dblfloatfix}

\usepackage{nag}

\makeatletter
\@ifpackageloaded{xcolor}{}{%
\usepackage[table,x11names,dvipsnames,svgnames]{xcolor}%
}
\makeatother

\usepackage{colortbl}

\usepackage{graphicx}
\usepackage{wrapfig}

\definecolorset{RGB}{lyft}{}{Red,194,39,36;Sunset,202,53,33;Orange,205,68,20;Amber,200,117,42;Yellow,242,169,52;Citron,186,188,44;Lime,112,159,33;Green,56,139,31;Mint,45,118,56;Teal,52,133,135;Cyan,60,132,202;Blue,55,94,248;Indigo,64,13,247;Purple,115,42,248;Pink,176,25,145;Rose,176,32,75}

\usepackage{cite}

\usepackage{microtype}

\usepackage{array}
\usepackage{multirow}
\usepackage{booktabs}
\usepackage{makecell} 

\ifcsname labelindent\endcsname
\let\labelindent\relax
\fi
\usepackage[inline]{enumitem}

\usepackage{subfig}

\newenvironment{lenumerate}[2][]
{\begin{enumerate}[label=(#2\arabic*),leftmargin=0.2in,itemindent=0.15in,#1]}
{\end{enumerate}}

\setlist*[enumerate,1]{label={\itshape\arabic*)}}

\makeatletter
\newcommand{\paragraphswithstop}{%
\let\copyparagraph\paragraph%
\renewcommand\paragraph[1]{\copyparagraph{##1.}}%
}
\makeatother

\usepackage[framemethod=tikz]{mdframed}

\usepackage{suffix}

\usepackage{environ}

\makeatletter
\newsavebox{\boxifnotempty}
\newcommand{\displayifnotempty}[3]{\sbox\boxifnotempty{#2}\setbox0=\hbox{\usebox{\boxifnotempty}\unskip}%
\ifdim\wd0=0pt
\else
 #1\usebox{\boxifnotempty}#3%
\fi%
}
\newcommand{\ifempty}[2]{\setbox0=\hbox{#1\unskip}%
\ifdim\wd0=0pt%
 #2%
\fi%
}
\newcommand{\ifnotempty}[2]{\setbox0=\hbox{#1\unskip}%
\ifdim\wd0>0pt%
 #2%
\fi%
}
\makeatother

\usepackage{algpseudocode}
\usepackage{algorithm}

\usepackage{scrlfile}

\makeatletter
\newcommand*\newstoreddef[1]{
  \BeforeClosingMainAux{%
    \immediate\write\@auxout{%
      \string\restoredef{#1}{\csname #1\endcsname}%
    }%
  }%
}
\newcommand*{\restoredef}[2]{
  \expandafter\gdef\csname stored@#1\endcsname{#2}%
}
\newcommand*{\storeddef}[1]{
  \@ifundefined{stored@#1}{0}{\csname stored@#1\endcsname}%
}
\makeatother

\usepackage{pageslts}
\NewEnviron{tee}{\BODY\typeout{Marker Tee [start] ^^J \BODY ^^JMaker Tee [end]}}

\input{preamble/math}
\newcommand{\real}[1]{\mathbb{R}^{#1}{}}

\newcommand{\naturals}[1]{\mathbb{N}^{#1}{}}

\newcommand{\bmat}[1]{\begin{bmatrix}#1\end{bmatrix}}

\DeclareMathOperator*{\argmax}{\arg\!\max}

\DeclareMathOperator{\clip}{{clip}}

\providecommand{\cM}{\mathcal{M}}

\usepackage{units}

\newcommand{\newcolorlabel}[2]{%
  \expandafter\newcommand\csname #1\endcsname[1]{%
    \colorbox{#2}{\color{white}\textsf{\textbf{##1}}}}%
}

\newcommand{\newcommenter}[2]{%
  \expandafter\newcommand\csname #1\endcsname[1]{%
    \fcolorbox{#2}{#2}{\color{white}\textsf{\textbf{#1}}}
    {\color{#2}##1}}%
  \expandafter\newcommand\csname at#1\endcsname{%
    \fcolorbox{#2}{#2}{\color{white}\textsf{\textbf{@#1}}}
    {\color{#2}}}%
  \expandafter\newcommand\csname #1hl\endcsname[2]{%
    \colorbox{#2}{\color{white}\textsf{\textbf{#1}}}\sethlcolor{Azure2}\hl{##2}~%
    \expandafter\ifx\csname commentarrow\endcsname\relax$\leftarrow$\else \commentarrow[#2]\fi~%
    {\color{#2}##1}}%
  \expandafter\newcommand\csname #1st\endcsname[2]{%
    \colorbox{#2}{\color{white}\textsf{\textbf{#1}}}\sout{##2}~%
    \expandafter\ifx\csname commentarrow\endcsname\relax$\leftarrow$\else \commentarrow[#2]\fi~%
    {\color{#2}##1}}%
}
\newcommenter{TODO}{DodgerBlue1}
\newcommenter{rtron}{Green3}

\usepackage{comment}

\usepackage{pdfcomment}

\usepackage{soul}

\usepackage[normalem]{ulem}

\usepackage{csquotes}

\usepackage{tikz}
\usetikzlibrary{calc}
\usetikzlibrary{matrix}
\usetikzlibrary{chains}
\usetikzlibrary{shapes.geometric}
\usetikzlibrary{arrows.meta}
\usetikzlibrary{decorations.pathreplacing}
\usetikzlibrary{backgrounds}

\tikzset{
  dim above/.style={to path={\pgfextra{
        \pgfinterruptpath
        \draw[>=latex,|->|] let
        \p1=($(\tikztostart)!1.5em!90:(\tikztotarget)$),
        \p2=($(\tikztotarget)!1.5em!-90:(\tikztostart)$)
        in(\p1) -- (\p2) node[pos=.5,sloped,above]{#1};
        \endpgfinterruptpath
      }
    }
  },
  dim double above/.style={to path={\pgfextra{
        \pgfinterruptpath
        \draw[>=latex,|->|] let
        \p1=($(\tikztostart)!3em!90:(\tikztotarget)$),
        \p2=($(\tikztotarget)!3em!-90:(\tikztostart)$)
        in(\p1) -- (\p2) node[pos=.5,sloped,above]{#1};
        \endpgfinterruptpath
      }
    }
  },
  dim below/.style={to path={\pgfextra{
        \pgfinterruptpath
        \draw[>=latex,|->|] let 
        \p1=($(\tikztostart)!-1em!-90:(\tikztotarget)$),
        \p2=($(\tikztotarget)!-1em!90:(\tikztostart)$)
        in (\p1) -- (\p2) node[pos=.5,sloped,below]{#1};
        \endpgfinterruptpath
      }
    }
  },
}

\tikzset{
    right angle quadrant/.code={
        \pgfmathsetmacro\quadranta{{1,1,-1,-1}[#1-1]}     
        \pgfmathsetmacro\quadrantb{{1,-1,-1,1}[#1-1]}},
    right angle quadrant=1, 
    right angle length/.code={\def\rightanglelength{#1}},   
    right angle length=2ex, 
    right angle symbol/.style n args={3}{
        insert path={
            let \p0 = ($(#1)!(#3)!(#2)$) in     
                let \p1 = ($(\p0)!\quadranta*\rightanglelength!(#3)$), 
                \p2 = ($(\p0)!\quadrantb*\rightanglelength!(#2)$) in 
                let \p3 = ($(\p1)+(\p2)-(\p0)$) in  
            (\p1) -- (\p3) -- (\p2)
        }
    }
}

\newcommand{\pgfextractangle}[3]{%
    \pgfmathanglebetweenpoints{\pgfpointanchor{#2}{center}}
                              {\pgfpointanchor{#3}{center}}
    \global\let#1\pgfmathresult  
}

\usetikzlibrary{shapes.arrows}
\newcommand{\commentarrow}[1][Azure4]{\tikz[baseline=-3pt]{\node[shape border uses incircle, fill=#1,rotate=180,single arrow, inner sep=1pt, minimum size=6pt, single arrow head extend=2pt]{};}}

\tikzset{ax/.style={-latex,line width=2pt}}

\tikzset{camera/.style={fill=Sienna1,fill opacity=0.5},%
image plane/.style={draw=RoyalBlue3,line width=2pt}}

\newcommenter{todo}{Firebrick1}
\newcommenter{danyang}{cyan}
\newcommenter{mingyu}{red}
\newcommenter{cristi}{orange}

\title{\LARGE \bf
Learning Signal Temporal Logic through Neural Network for Interpretable Classification
}

\author{Danyang Li$^{1}$, Mingyu Cai$^{2}$, Cristian-Ioan Vasile$^{2}$, Roberto Tron$^{1}$
\thanks{$^{1}$Danyang Li and Roberto Tron are with Mechanical Engineering Department, Boston University, Boston, MA 02215, USA.
        {\tt\small danyangl@bu.edu, tron@bu.edu}}%
\thanks{$^{2}$M. Cai and C.I. Vasile are with the Department of Mechanical Engineering and Mechanics, Lehigh University, Bethlehem, PA, USA
        {\tt\small mingyu-cai@lehigh.edu, cvasile@lehigh.edu}}%
}

\newtheorem{example}{Example}
\usepackage{tikz}
\usepackage{etoolbox}
\usepackage{listofitems}
\tikzset{>=latex} 
\usepackage{threeparttable}
\usepackage{cleveref}

\begin{document}

\def\next{\circ}
\def\always{\square}
\def\event{\lozenge}
\def\relu{\text{ReLU}}
\def\layersep{2.5}
\def\layerseps{2.7}
\def\yshift{-2}

\maketitle
\thispagestyle{empty}
\pagestyle{empty}

\begin{abstract}
Machine learning techniques using neural networks have achieved promising success for time-series data classification. However, the models that they produce are challenging to verify and interpret.
In this paper, we propose an explainable neural-symbolic framework for the classification of time-series behaviors. In particular, we use an expressive formal language, namely Signal Temporal Logic (STL), to constrain the search of the computation graph for a neural network. We design a novel time function and sparse softmax function to improve the soundness and precision of the neural-STL framework. As a result, we can efficiently learn a compact STL formula for the classification of time-series data through off-the-shelf gradient-based tools.
We demonstrate the computational efficiency, compactness, and interpretability of the proposed method through driving scenarios and naval surveillance case studies, compared with state-of-the-art baselines.
\end{abstract}

\section{INTRODUCTION}
Binary classification is a standard task in Supervised Machine Learning (ML). Many techniques such as K-Nearest Neighbors (KNN), Supported Vector Machine (SVM), and Random Forest can be employed. A particular version of the procedure introduces data derived from the evolution of dynamic systems, which are in the form of time-series data; this task can still be challenging for the current state-of-the-art ML algorithm.

In recent years, deep learning (DL) methods have achieved great success in complex pattern recognition, with a variety of architectures such as convolutional~\cite{acharya2017deep}, recurrent~\cite{tang2015document}, etc. However, the common drawback of DL and other ML techniques is the lack of human-interpretability of the learned models. This property is very important in control and robotics since it can support qualitative and quantitative analysis. In this paper, we propose to use the rich expressivity of Signal Temporal Logic (STL) ~\cite{maler2004monitoring} for training classification of symbolic neural networks.

Signal Temporal Logic (STL) is a formal language designed to express rich and human-interpretable specifications over time series. Its quantitative satisfaction can be measured via a recursive definition of robustness~\cite{maler2004monitoring}, which can be used in control optimization problems~\cite{raman2014model} and sequential predication task~\cite{ma2020stlnet}. The task of learning STL formulas for classification also leverages the robustness to formulate an optimization problem. Previous attempts~\cite{jin2013mining, bakhirkin2018efficient, hoxha2018mining} for learning temporal logic properties assumed fixed formula structures. Decision tree based approaches~\cite{kong2016temporal,bombara2016decision, mohammadinejad2020interpretable, aasi2021classification,gaglione2022maxsat} are proposed to learn both operators and parameters via growing logic branches. The branches growing of decision trees increases as the structure of desired STL formula becomes more complicated. To the best of our knowledge, the training time for the decision tree based method is in order of minutes while the temporal logic neural network can achieve order of seconds.
To learn temporal logic formulas from scratch, developing logic neural networks~\cite{rocktaschel2017end, yang2017differentiable} attracts more attention.
However, the max/min operations in the computation of robustness are not continuously differentiable. These operations cannot be directly applied in a neural network since the network is not able to learn with it.

\textbf{Related works:\space}
Many works use STL for classification. Some works~\cite{leung2020back, ketenci2020learning, yan2021neural,baharisangari2021weighted,chen2022interpretable} use the recursive definition of the STL and the structure of multi-layer neural networks, allowing to directly apply
off-the-shelf gradient-based deep learning software.
This kind of work can be divided into two categories i.e., template-free and template-based learning.
Intuitively, template-based learning~\cite{leung2020back, ketenci2020learning, yan2021neural,baharisangari2021weighted} fixes the STL formula to be learned from data and only learns some parameters. Whereas template-free learning~\cite{chen2022interpretable} learns the overall STL formula including its structure from scratch. Our work uses this latter category.

To address the non-smooth issues of the original STL robustness\cite{maler2004monitoring}, the differentiable versions of robustness computation were proposed in~\cite{leung2020back,yan2021neural}, which, however, can not guarantee the soundness properties of STL and the accuracy of satisfaction. The work~\cite{chen2022interpretable} proposes an alternative formulation that uses arithmetic-geometric mean (AGM) STL robustness~\cite{mehdipour2019arithmetic}. The AGM robustness is not continuous and can not directly be backpropagated in a neural network, resulting in problems of efficiency and adaptation to gradient-based tools. In addition, both works~\cite{yan2021neural,chen2022interpretable} learn a weighted STL (wSTL)~\cite{9309020} that is harder for a user to interpret if the complex weights are not manually omitted.

\textbf{Contributions:} First, we propose a template-free neural network architecture with given fragments for the classification of time-series signals. The network can be interpreted as level-one STL formulae. Our proposed models rely on novel activation functions, especially on two new components: a compact way to represent time intervals for temporal operators, and a sound, differentiable approximation to the maximum operator. As an additional contribution, we introduce a synthetic dataset containing prototypical behaviors inspired by autonomous driving situations. We use this dataset to evaluate the performance of our methods, showing that we are able to achieve high efficiency and automatically extract qualitatively meaningful features of the data. More broadly, we hope this dataset will serve as a benchmark for future work in the classification of time series.

\section{Preliminaries and Problem Formulation}\label{sec-2}
We denote a $d$-dimensional, discrete-time, real-valued signal as $s=(s(0),s(2),...,s(l))$, where $l\in \naturals{+}$ is the length of the signal, $s(\tau)\in\real{d}$ is the value of signal $s$ at time $\tau$.
\subsection{Signal Temporal Logic}
We use signal temporal logic (STL) to specify the behaviors of time-series signals.
\begin{definition}
The syntax of STL formulae with linear predicates is defined recursively as~\cite{maler2004monitoring}:
\begin{equation}
    \phi::=\mu \mid \neg \phi \mid \phi_1\land\phi_2 \mid \phi_1\lor\phi_2 \mid \event_{[t_1,t_2]}\phi \mid \always_{[t_1,t_2]}\phi,
\end{equation}
where $\mu$ is a predicate defined as $\mu:=\mathbf{a}^Ts(\tau)\geq b$, $\mathbf{a}\in\real{d}$, $b\in\real{}$. $\phi,\phi_1,\phi_2$ are STL formulas. The Boolean operators $\neg,\land,\lor$ are \emph{negation}, \emph{conjunction} and \emph{disjunction}, respectively. The temporal operators $\event,\always$ represent \emph{eventually} and \emph{always}. $\event_{[t_1,t_2]}\phi$ is true if $\phi$ is satisfied for at least one point $\tau\in[t_1,t_2]$, while $\always_{[t_1,t_2]}\phi$ is true if $\phi$ is satisfied for all time points $\tau\in[t_1,t_2]$.
\end{definition}

\begin{definition}\label{def:r}
The quantitative semantics, i.e., the robustness, of the STL formula $\phi$ over signal $s$ at time $\tau$ is defined as:
\begin{subequations}\label{stl-semantics}
\begin{align}
    r(s,\mu,\tau) &=\mathbf{a}^Ts(\tau)-b\label{stl-semantics:sub1},\\
    r(s,\neg \phi,\tau) &=-r(s,\phi,\tau)\label{stl-semantics:sub2},\\
    r(s,\land_{i=1}^N \phi_i,\tau) &= \min(\{r(s,\phi_i,\tau)\}_{i=1:N})\label{stl-semantics:sub3},\\
    r(s,\lor_{i=1}^N \phi_i,\tau) &= \max(\{r(s,\phi_i,\tau)\}_{i=1:N})\label{stl-semantics:sub4},\\
    r(s,\always_{[t_1,t_2]}\phi,\tau) &= \min_{\tau'\in[\tau+t_1,\tau+t_2]}r(s,\phi,\tau')\label{stl-semantics:sub5},\\
    r(s,\event_{[t_1,t_2]}\phi,\tau) &= \max_{\tau'\in[\tau+t_1,\tau+t_2]}r(s,\phi,\tau')\label{stl-semantics:sub6},
\end{align}
\end{subequations}
The signal $s$ is said to satisfy the formula $\phi$,
denoted as $s \models \phi$, if $r(s, \phi, 0) > 0$.
Otherwise, $s$ is said to violate $\phi$, denoted as $s \not\models \phi$.
\end{definition}

\begin{definition}\label{def:r-ind}
The quantitative semantics in Definition~\ref{def:r} can be expressed in the form of indicator vectors:
\begin{subequations}\label{stl-indicator}
\begin{align}
    r(s,\mu,\tau) &=\mathbf{a}^Ts-b\label{stl-ind:sub1},\\
    r(s,\neg \phi,\tau) &=-r(s,\phi,\tau)\label{stl-ind:sub2},\\
    r(s,\land_{i=1}^N \phi_i,\tau) &= -f(-\mathbf{r}_c,\mathbf{c})\label{stl-ind:sub3},\\
    r(s,\lor_{i=1}^N \phi_i,\tau) &= f(\mathbf{r}_d,\mathbf{d})\label{stl-ind:sub4},\\
    r(s,\always_{\mathbf{w_t}}\phi,\tau) &= -f(-\mathbf{r}_t,\mathbf{w_t})\label{stl-ind:sub5},\\
    r(s,\event_{\mathbf{w_t}}\phi,\tau) &= f(\mathbf{r}_t,\mathbf{w_t})\label{stl-ind:sub6},
\end{align}
\end{subequations}
where $\mathbf{r}_c=\text{stack}(\{r(s,\phi_i,\tau)\}_{i=1;N})$, $\mathbf{r}_d=\text{stack}(\{r(s,\phi_i,$ $\tau)\}_{i=1;N})$, $\mathbf{r}_t=\text{stack}(\{r(s,\phi,\tau')\}_{\tau'\in[0,l]})$. $\mathbf{c}$, $\mathbf{d}$ and $\mathbf{w_t}$ are binary indicator vectors. $\mathbf{w_t}=\bmat{w_{t0},...,w_{ti},...,w_{tl}}$ is based on the time parameters in \ref{stl-semantics:sub5}, \ref{stl-semantics:sub6}. In other words, $w_{ti}=1$ if $i\in[\tau+t_1,\tau+t_2]$, otherwise, $w_{ti}=0$. $f(\mathbf{r},\mathbf{w})=\max\{\mathbf{r_s}\}$, where $\mathbf{r_s}$ is obtained from $\mathbf{r}$ using the indicator vector $\mathbf{w}$ such that $r_i\in\mathbf{r_s}$ if $w_i=1$ and $r_i\notin\mathbf{r_s}$ if $w_i=0$.
\end{definition}

\begin{definition}[Soundness]\label{def:sound}
Let $\cM$ denote an algorithm for computing the robustness of an STL formula $\phi$ over signal $s$. We say $\cM$ is sound if $\cM(\phi,s)$ can provide a valid result for evaluating STL satisfaction, i.e.,
\begin{subequations}
\begin{align}
    s \models \phi \Longrightarrow \cM(\phi,s)>0\\
    s \not\models \phi \Longrightarrow \cM(\phi,s)\leq 0
\end{align}
\end{subequations}
\end{definition}

\subsection{Problem Statement}
\label{sec:problem-formulation}
We consider a binary classification task on a labeled data set $S=\{(s^i,c^i)\}_{i=1}^N$, where $s^i\in\real{l\times d}$ is the $i^{th}$ signal labeled as $c^i\in C$. $C=\{1,-1\}$ is the set of classes.

\begin{problem}
Given $S=\{(s^i,c^i)\}_{i=1}^N$, find an STL formula $\phi$ that minimizes the misclassification rate $MCR$ defined as:
\begin{equation*}
    MCR = \frac{\lvert\{s^i\mid (s^i \models \phi \land c^i = -1) \lor (s^i \not\models \phi \land c^i = 1)\}\rvert}{N}
\end{equation*}
\end{problem}

\smallskip

\begin{figure*}[tb]
\centering
\scalebox{.75}{
\begin{tikzpicture}[shorten >=2pt,->,draw=black!50, node distance=\layersep, yshift=\yshift]
    \tikzstyle{every pin edge}=[<-,shorten <=1pt]
    \tikzstyle{neuron}=[circle,fill=black!25, minimum size=1.2cm,inner sep=0pt]
    \tikzstyle{op}=[rectangle,fill=black!25, minimum size=1.2cm,inner sep=0pt, rounded corners]
    \tikzstyle{predicate neuron}=[neuron, fill=green!60!black!40];
    \tikzstyle{spatial operator neuron}=[op, fill=red!50];
    \tikzstyle{time}=[op, fill=orange!50];
    \tikzstyle{conjunction}=[neuron, fill=blue!40!];
    \tikzstyle{disjunction}=[neuron, fill=yellow!40!];
    \tikzstyle{annot} = [text width=4em, text centered];
    
    \node(rectangle) [minimum width = 2cm, minimum height = 7cm, draw=black, fill=white, rounded corners] (Operator) at (2*\layerseps,\yshift-3.4) {};
    \node[time] (t) at (\layersep,-1) {$\mathbf{t_1},\mathbf{t_2}$};
    \node[predicate neuron] (F-1) at (0,\yshift-1.3) {$\pi(s(0))$};
    \node[predicate neuron] (F-2) at (0,\yshift-1.3*2) {$\pi(s(1))$};
    \node at (0,\yshift-1.3*3) {$\vdots$};
    \node[predicate neuron] (F-3) at (0,\yshift-1.3*4) {$\pi(s(l))$};
    \node[predicate neuron] (P-1) at (\layersep,\yshift-1.3) {\large$r_0$};
    \node[predicate neuron] (P-2) at (\layersep,\yshift-1.3*2) {\large$r_1$};
    \node at (\layersep,\yshift-1.3*3) {$\vdots$};
    \node[predicate neuron] (P-3) at (\layersep,\yshift-1.3*4) {\large$r_l$};
    \node[spatial operator neuron] (S-1) at (2*\layerseps,\yshift+0.5-1.3*1 ){\large$g^{\always}_1$};
    \node[spatial operator neuron] (S-2) at (2*\layerseps,\yshift+0.45-1.3*2 ) {\large$g^{\event}_2$};
    \node at (2*\layerseps,\yshift+0.5-1.3*3) {$\vdots$};
    \node[spatial operator neuron] (S-3) at (2*\layerseps,\yshift+0.45-1.3*4 ){\large$g^{\always}_{m-1}$};
    \node[spatial operator neuron] (S-4) at (2*\layerseps,\yshift+0.5-1.3*5 ) {\large$g^{\event}_m$};
    \node[predicate neuron] (P-1) at (\layersep,\yshift-1.3) {\large$r_0$};
    \node[predicate neuron] (P-2) at (\layersep,\yshift-1.3*2) {\large$r_1$};
    \node at (\layersep,\yshift-1.3*3) {$\vdots$};
    \node[predicate neuron] (P-3) at (\layersep,\yshift-1.3*4) {\large$r_l$};
    \node[conjunction] (C-1) at (3*\layerseps,\yshift-1.3) {\large$\land_1$};
    \node[conjunction] (C-2) at (3*\layerseps,\yshift-1.3*2) {\large$\land_2$};
    \node at (3*\layerseps,\yshift-1.3*3) {$\vdots$};
    \node[conjunction] (C-3) at (3*\layerseps,\yshift-1.3*4) {\large$\land_k$};
    \node[disjunction] (D-1) at (4*\layerseps,\yshift-1.3*2-0.6) {\large$\lor$};

    \draw[->] (t) -| (Operator.north);
    \draw [<-] (F-1) -- ++(-1,0) node [left] {$s(0)$};
    \draw [<-] (F-2) -- ++(-1,0) node [left] {$s(1)$};
    \draw [<-] (F-3) -- ++(-1,0) node [left] {$s(l)$};
    \foreach \s/ \d in {1,...,3}
        \draw [->] (F-\s) -- (P-\d);
    \foreach \s in {1,...,3} {
      \draw[->] (P-\s) -- (P-\s-|Operator.west);}
    \foreach \s in {1,...,4}
        \foreach \d in {1,...,3}
            \draw [->] (S-\s) -- (C-\d);
    \draw [->] (S-1) -- (C-1) node[midway,above]{$c_{11}$};
    \draw [->] (S-4) -- (C-3) node[midway,below]{$c_{mk}$};
    \foreach \s in {1,...,2}
        \foreach \d in {1}
            \draw [->] (C-\s) -- node [above,midway] {$d_{\s}$} (D-\d);
    \draw [->] (C-3) -- node [above,midway] {$d_n$} (D-1);
    \foreach \s in {1}
        \draw [->] (D-\s) -- ++(1,0) node [right] {$r_{\chi}$};
\end{tikzpicture}
}
\caption{The structure of NN-TLI. The temporal operators accept predicates $r_0,...,r_l$ and time variables $\mathbf{t_1},\mathbf{t_2}$ as inputs. Here, $\mathbf{t_1}$ and $\mathbf{t_2}$ are vectors containing the starting time variables and ending time variables for all temporal operators.}
\label{fig:structure}
\end{figure*}
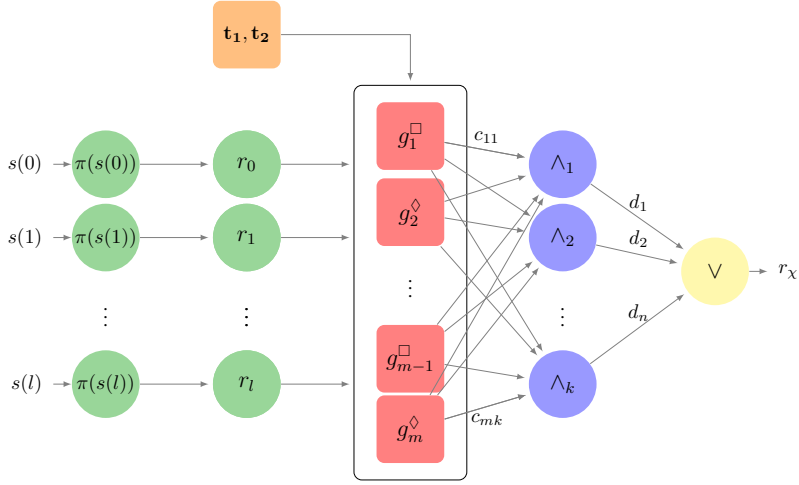

\subsection{Overview of Proposed Work}
\label{sec:overview}

We propose a Temporal Logic Inference (TLI) method based on neural network embedding of STL formulas.
Using our method, we can infer both the structure and parameters of STL properties.
The proposed Neural Network architecture for TLI (NN-TLI) enables the use of high-performance training methods and extraction of interpretable classifiers in the form of STL formulas from networks' weights.
The NN-TLI maps a time-series data point (a signal) $s$ to its robustness value of the learned formula with the sign of robustness defining the predicted class of signal $s$.
We restrict ourselves to an expressive fragment of STL that allows the proposed efficient encoding
and that is easy to interpret. 
In the following, we present the NN-TLI structure in Sec.~\ref{sec:nn-tli-arch} and its training method in Sec.~\ref{sec:nn-tli-training}.

\section{Neural Network Architecture}
\label{sec:nn-tli-arch}
The overall architecture of NN-TLI consists of four layers. The first layer is associated with the predicates of sub-formulas. It takes signals as input and outputs robustness values for all predicates at all times.
The second layer takes robustness values from the first layer and time variables as input and applies temporal operators to generate sub-formulas.
The last two layers implement logical operators to combine sub-formulas. In the following sections, we first introduce the fragment of STL captured by NN-TLI formulas.
Next, we describe details of each layer of the NN-TLI architecture and show how they implement the STL formula.

\subsection{STL(1)}

We define level-one STL, denoted by STL$(1)$, as
\begin{equation}
\begin{aligned}
    \phi&:=\mu \mid \neg \phi \mid \phi_1\land\phi_2 \mid \phi_1 \lor \phi_2,\\
    \psi&:=\event_{[t_1,t_2]}\phi|\always_{[t_1,t_2]}\phi,
\end{aligned}
\end{equation}

It is a proper fragment of STL without nested temporal operators; it still retains rich expressivity as shown by the following result.

\begin{theorem}
    The STL fragment STL$(1)$ is equal to its Boolean closure~\cite{bombara2016decision}.
\end{theorem}
\begin{proof}[Sketch]
    It follows immediately from the distributivity properties of always and eventually with respect to Boolean operators, i.e.,
    $\event_I (\phi_1 \lor \phi_2) = \event_I \phi_1 \lor \event_I \phi_2$,
    $\always_I (\phi_1 \land \phi_2) = \always_I \phi_1 \land \always_I \phi_2$,
    and $\lnot \always_I \phi = \event_I \lnot \phi$.
\end{proof}

Note, STL$(1)$ captures the same fragment as DT-STL~\cite{bombara2016decision}. However, the representation of formulas is more compact since we do not impose the decision tree structure with primitive formula decision nodes.

In this work, we consider STL$(1)$ formulas in the disjunctive normal form (DNF), which can be expressed as
\begin{equation}
\label{eq:stl_formula}
    \chi = (\psi_{11}\land\psi_{12}\land...\land\psi_{1p})\lor...\lor(\psi_{n1}\land\psi_{n2} \land...\land\psi_{nq}),
\end{equation}
where $\psi_{ij} :=\event_{[t_1,t_2]} \varphi_{ij}\mid \always_{[t_1,t_2]} \varphi_{ij}$, $\varphi_{ij}:=\mu\mid\neg\mu$. The STL formula $\chi$ takes disjunction of $n$ terms, $p,...,q$ are the number of sub-formulas to take conjunction inside each term.

We use a matrix called \emph{conjunction-disjunction matrix} $M$ to represent DNF in \eqref{eq:stl_formula}. The conjunction-disjunction matrix is a binary matrix defined as:
\begin{equation}\label{eq:c-d matrix}
    M = \bmat{\mathbf{c_1}\\\mathbf{c_2}\\\vdots\\\mathbf{c_m}}\\
    = \bmat{c_{11} &c_{12} &... &c_{1k}\\
    c_{21} &c_{22} &... &c_{2k}\\
    \vdots &\vdots &\ddots &\vdots\\
    c_{m1} &c_{m2} &... &c_{mk}\\}\in\{0,1\}^{m\times k}
\end{equation}
where $m>n$, $k>\max\{p,...,q\}$. $\mathbf{c}_i$ is the indicator vector in \eqref{stl-ind:sub3}. We can define the indicator vector $\mathbf{d}$ in \eqref{stl-ind:sub4} based on $M$ as:
\begin{equation}\label{eq:ber1}
\begin{aligned}
\mathbf{d} &= \bmat{d_1,d_2,...,d_m}^T\\
    d_i &=
    \begin{cases}
    0 & \textrm{if } \mathbf{c}_i=\mathbf{0},\\
    1 & \textrm{otherwise.}
    \end{cases}
\end{aligned}
\end{equation}
The STL(1) formula can be easily constructed from the conjunction-disjunction matrix. For example, we learn 5 sub-formulas $\phi_1, \phi_2, \phi_3, \phi_4, \phi_5$ and a conjunction-disjunction matrix $M=\bmat{1 &1 &0 &0 &0\\ 0 &1 &1 &1 &0\\ 0 &0 &0 &0 &0}$, then the final STL(1) formula is $\psi = (\phi_1\land\phi_2)\lor(\phi_2\land\phi_3\land\phi_4)$. The architecture of NN-TLI in the form of \eqref{eq:stl_formula} is shown in Fig~\ref{fig:structure}.

\begin{theorem}
A logical formula can be converted into an equivalent DNF.\cite{davey2002introduction}
\end{theorem}
\begin{remark}
The maximum size of the template is fixed, whereas the exact size of the template is free to learn.
\end{remark}
\begin{remark}\label{remark:dnf}
All level-one STL formulae can be encoded in an NN-TLI of sufficient size.
\end{remark}

\subsection{First layer: Spatial Predicates}

The first layer generates the predicates in the form:
\begin{equation}
    \mu=\pi(s(\tau))=\mathbf{a}^Ts(\tau) - b,
\end{equation}
where $\mathbf{a}$ or $-\mathbf{a}$ is an element of the standard basis of the space $\real{d}$, $b\in\real{}$. Note that $\mathbf{a}$ is fixed, and $b$ is a variable. The output of $\pi(s(\tau))$ is the axis-aligned quantitative satisfaction of signal $s$ at time $\tau$.

\subsection{Second and third layer: temporal and logical operators}

The second layer has 2 types of temporal operators, "eventually" $\event_{[t_1,t_2]}$ and "always" $\always_{[t_1,t_2]}$, where the time variables $t_1, t_2$ are learned from data. The third layer includes 2 types of logical operators, conjunction $\land$ and disjunction $\lor$. 
The quantitative semantics of these four operators are defined in \eqref{stl-semantics}. However, \eqref{stl-semantics:sub3}, \eqref{stl-semantics:sub4} are not differentiable with respect to time variables $t_1,t_2$. Moreover, the gradient of the min/max function in \eqref{stl-semantics} is zero everywhere except for the min/max point, thus, the neural network does not learn through the back-propagation.

In this paper, we propose a novel definition of robustness for operators $\always_{[t_1,t_2]},\event_{[t_1,t_2]},\land,\lor$; which are differentiable with respect to variables $t_1,t_2$ and $s$. 
The activation functions are defined as:
\begin{subequations}
\begin{align}
g^{\always}(\mathbf{r},t_1,t_2;\eta,\beta,h) &= -F(-\mathbf{r},f_t(t_1,t_2;\eta);\beta,h),\label{always}\\
g^{\event}(\mathbf{r},t_1,t_2;\eta,\beta,h) &=F(\mathbf{r}, f_t(t_1,t_2;\eta);\beta,h),\label{eventually}\\
g^{\land}(\mathbf{r},\mathbf{c}^g;\beta,h) &= -F(-\mathbf{r},p(\mathbf{c}^g);\beta,h),\label{and}\\
g^{\lor}(\mathbf{r},\mathbf{d};\beta,h) &= F(\mathbf{r},\mathbf{d};\beta,h) \label{or},
\end{align}
\end{subequations}
where $\eta,\beta,h$ are hyperparameters, $F(\cdot)$ is the \emph{sparse softmax function} (See \ref{soft}), $f_t(\cdot)$ is the \emph{time function} (See \ref{time}), $\mathbf{c}^g$ and $\mathbf{d}$ are vectors from \emph{conjunction-disjunction gate matrix} (See \ref{logical}) and $p(\cdot)$ is the \emph{gate function} (See \ref{logical}).

In the following sections, we will introduce details of the operator layers and provide rigorous soundness analysis for the proposed STL robustness.

\subsection{Sparse softmax function}\label{soft}
Since the standard robustness computation in Definition~\ref{def:r} including $\max$ and $\min$ is non-smooth and not differentiable, the alternative solution is to use softmax for the approximation:
\begin{equation}\label{softmax}
\begin{aligned}
    s(\mathbf{r},\mathbf{w};\beta)=\frac{\sum_{i=0}^{l} r_i w_i e^{\beta r_i}}{\sum_{i=0}^{l} w_i e^{\beta r_i}}= \sum_{i=0}^{l} r_i q^s_i,
\end{aligned}
\end{equation}
where $\mathbf{r}$ is an input vector, $\mathbf{w}$ is a weight vector, $q^s_i=\frac{w_i e^{\beta r_i}}{\sum_{i=0}^{l} w_i e^{\beta r_i}}$. This function can also be used to approximate $\min$ by simply replacing the $e^{\beta r_i}$ with $e^{-\beta r_i}$. Such a calculation has been applied as activation functions of temporal-logic-based neural networks in~\cite{leung2020back,yan2021neural}. However, the softmax approximation always introduces errors in the estimated robustness. As a result, soundness properties of STL robustness might not be guaranteed~\cite{maler2004monitoring}, leading to inaccurate results and erroneous STL interpretation. For example, for a learned STL formula $\phi$, the true robustness computed from \eqref{stl-semantics} given a signal $s$ is $r(s,\phi)$ and the approximated robustness using softmax function is $\tilde{r}(s,\phi)$. It is possible to have $r(s,\phi)<0$ while $\tilde{r}(s,\phi)>0$, then the learned STL formula cannot interpret the property of signal $s$ correctly since $s \not\models \phi$.

To address the issue, we present \emph{sparse softmax function} that can approximate the maximum of the subset of elements, while still guaranteeing STL soundness properties. Our sparse softmax function $F(\mathbf{r},\mathbf{w};\beta,h)$ is defined through the following sequence of operations:
\begin{subequations}\label{sparsemax}
\begin{align}
    r_i' &= r_i w_i\label{max:sub1}\\
    r_{im} &= \begin{cases}
    \lvert \max_{i}(r_i')\rvert & \textrm{if } \lvert \max_i(r_i')\rvert \neq 0,\\
    1 & \textrm{otherwise.}
    \end{cases}\label{max:rm}\\
    r_i'' &= \frac{h r_i'}{r_{im}}\label{max:sub2},\\
    q_i &= \frac{e^{\beta r_i''}}{\sum_i e^{\beta r_i''}}\label{max:sub3},\\
    F(\mathbf{r},\mathbf{w};\beta,h) &= \frac{\sum_{i=0}^{l} r_i w_i q_i}{\sum_{i=0}^{l} w_i q_i} = \sum_{i=0}^{l}r_i q^F_i \label{max:sub4},
\end{align}
\end{subequations}
where $h\in\real{+}$ and $\beta\in\real{}$ are hyperparameters, $q^F_i=\frac{w_i q_i}{\sum_{i=1}^{n} w_i q_i}$, $\mathbf{r}$ contains elements to take for maximum, $\mathbf{w}$ is the binary indicator vector defined in Definition~\ref{def:r-ind}, from which we can obtain a subset of elements in $\mathbf{r}$.

The approximated maximum $F(\mathbf{r},\mathbf{w};\beta,h)$ in \eqref{max:sub4} is the weighted sum of $r_i's$. We first select the elements to take for maximum by zeroing out $r_i$ if $w_i=0$ in \eqref{max:sub1}. Next, we scale $r_i'$ to $r_i''$ in \eqref{max:sub2}. We transform the result into a probability max function $q$ in \eqref{max:sub3}.

\begin{proposition}
The sparse softmax function $F(\mathbf{r},\mathbf{w};\beta,h)$ defined in \eqref{sparsemax} is sound, i.e.,
\begin{subequations}
    \begin{align}
        f(\mathbf{r},\mathbf{w})\leq 0 &\Longrightarrow F(\mathbf{r},\mathbf{w};\beta,h) \leq 0\label{smax1},\\
        f(\mathbf{r},\mathbf{w})>0 &\Longrightarrow F(\mathbf{r},\mathbf{w};\beta,h) >0\label{smax2},\\
        -f(-\mathbf{r},\mathbf{w})\leq 0 &\Longrightarrow -F(-\mathbf{r},\mathbf{w};\beta,h) \leq 0\label{smin1},\\
        -f(-\mathbf{r},\mathbf{w})>0 &\Longrightarrow -F(-\mathbf{r},\mathbf{w};\beta,h) >0\label{smin2},
    \end{align}
\end{subequations}
if the hyperparameters $\beta$, $h$ satisfies $he^{\beta h}>\frac{l e^{-1}}{\beta}$ , where $\mathbf{r}\in\real{l+1}$. $f(\mathbf{r},\mathbf{w})$ is defined in \eqref{stl-indicator}.
\end{proposition}

\begin{proof}
    First, note that $q_i$ is always positive and $w_i$ is always non-negative. If $\mathbf{w}$ is a zero vector, from the property of indicator vector, $f(\mathbf{r},\mathbf{w})$ is an empty set, the robustness is meaningless. Thus, we only focus on the cases when $\mathbf{w}$ is a non-zero vector, then the denominator in \eqref{max:sub4} $\sum_{i=0}^{l} w_i q_i$ is always positive.
    
    We start from case \eqref{smax1}, from LHS, $r_i w_i \leq 0$, $\forall i\in[0,l]$, then from \eqref{max:sub4}$, F(\mathbf{r},\mathbf{w};\beta,h)\leq 0$ is derived. Similarly, for case \eqref{smin2}, from LHS, $r_iw_i\geq 0$, $\forall i\in[0,l]$ and $r_jw_j> 0$, $\exists j\in[0,l]$, from \eqref{max:sub4}, $F(\mathbf{r},\mathbf{w};\beta,h)>0$ is derived.
    
    For case \eqref{smax2}, we know that $\sum_{i=0}^{l} w_i q_i>0$, from \eqref{max:sub4}, we need $\sum_{i=0}^{l} r_i w_i q_i>0$ to achieve soundness. Let $k=\argmax_i r_iw_i$, from LHS of \eqref{smax2}, $r_k w_k>0$, then $r_k>0$, $w_k=1$. Since $\frac{h}{r_{im}}$ is positive, $\sum_{i=0}^{l} r_i w_i q_i>0\iff\sum_{i=0}^{l}\frac{h r_i w_i}{r_{im}}q_i\iff\sum_{i=0}^{l} r_i'' q_i>0$. We can get $r_k''=h$ from \eqref{max:sub2}, $\sum_{i=0}^{l} r_i'' q_i = r_k'' q_k+\sum_{i=0, i\neq k}^{l} r_i'' q_i=\frac{r_k'' e^{\beta r_k''}+\sum_{i=0, i\neq k}^{l} r_i'' e^{\beta r_i''}}{\sum_{i=0}^{l} e^{\beta r_i''}}$. Since $\sum_{i=0}^{l} e^{\beta r_i''}>0$, we need $r_k'' e^{\beta r_k''}+\sum_{i=0, i\neq k}^{l} r_i'' e^{\beta r_i''}>0$ to achieve soundness. The minimum of function $g(x)=xe^{\beta x}$ is $-\frac{e^{-1}}{\beta}$, thus the minimum value of ${r_i'' e^{\beta r_i''}}$ is $-\frac{e^{-1}}{\beta}$, we have $r_k'' e^{\beta r_k''}+\sum_{i=0, i\neq k}^{l} r_i'' e^{\beta r_i''}\geq he^{\beta h}-\frac{l e^{-1}}{\beta}>0$.
    
    The proof for \eqref{smin1} can be derived from case \eqref{smax2} by multiplying $-1$ on both sides of \eqref{smin1} and replacing $r_i$ with $-r_i$ in \eqref{smax2}.
\end{proof}

\begin{example}
Let a signal $s=2,1.1,1,0,-1$. Consider the specification $\phi=\event_{[1,4]} (s>1)$. The indicator vector $\mathbf{w}=\bmat{0 &1 &1 &1 &1}$ from the time interval $[1,4]$. The sequence of robustness values of predicate $\mu =s-1$ is $\mathbf{r}=\bmat{1 &0.1 &0 &-1 &-2}$. The true robustness computed from \eqref{stl-semantics} is $r(s,\phi)=0.1>0$. Choosing $\beta=1$, the approximated robustness computed from traditional softmax function in \eqref{softmax} is 
\begin{equation}
\begin{aligned}
    s(\mathbf{r},\mathbf{w};\beta)=\sum_{i=1}^{4} x_i q^s_i = -0.202<0\\
\end{aligned}
\end{equation}
where $\mathbf{q^s} = \bmat{0&0.43&0.38&0.14&0.05}$.

The approximated robustness using our sparse softmax function in \eqref{sparsemax} is  
\begin{equation}
\begin{aligned}
    \mathbf{r}' &= \bmat{0 &0.1 &0 &-1 &-2},\\
    \mathbf{r}'' &= \bmat{0 &1 &0 &-10 &-20},\\
    \mathbf{q^F} &= \bmat{0 &0.73 &0.27 &1.2\cdot 10^{-5} &5.5\cdot 10^{-10}}\\
    F(\mathbf{r},\mathbf{w};\beta,h) &= \sum_{i=0}^{l-1} x_i q^F_i = 0.073>0,
\end{aligned}
\end{equation}
with $h=1$ to satisfy $he^{\beta h}>\frac{4e^{-1}}{\beta}$. With the temporal information from indicator vector $\mathbf{w}$ and the rescaling from \cref{max:sub1,max:rm,max:sub2}, our sparse softmax function can provide valid robustness while the traditional softmax function cannot. When it comes to classification, the algorithm using traditional softmax function will classify the signal $s$ as label $-1$, whereas the signal s satisfies $\phi$.
\end{example}

We can take advantage of the sparse softmax function from the following aspects:
\begin{enumerate}
    \item Soundness guarantees allow the learned STL formula of NN-TLI to provide valid results, thus can improve the accuracy and interpretability of the learned STL formula for classifying time-series data.
    
    \item The function is able to approximate the maximum/minimum of a subset of elements.
    
    \item The function is differentiable and enables the application of off-the-shelf machine learning tools. We can take advantage of their large collection of optimization tools and accelerate the training process by parallelizing training.
\end{enumerate}

\subsection{Time function}\label{time}

Learning the time interval of an STL formula can be challenging since the number of parameters of indicator vector $\mathbf{w_t}$ in \eqref{stl-indicator} is linear in the length of the signal, despite the fact that an interval is always specified by time indices. Moreover, learning the independent time weight might produce multiple time windows for one temporal operator, making it hard to construct the final STL formula. Previous works used fixed time windows~\cite{leung2020back,yan2021neural} or an autoencoder~\cite{chen2022interpretable} to predict a time interval for every signal; in the latter, the time window (and hence the final formula) will, in general, differ on the signal.

We propose a differentiable and efficient \emph{time function} to generate a binary indicator vector, called time vector, to encode the time interval $[t_1, t_2]$. This time function can be considered as a smooth indicator function and only requires two parameters $t_1,t_2$ regardless of the length of the signal. The time function $f_t$ is defined as:
\begin{equation}
\begin{aligned}
&f_t(t_1,t_2;\eta)\\
=&\frac{1}{\eta}\min \Bigl(\bigl(\relu(\mathbf{n}-\mathbf{1}(t_1-\eta))-\relu(\mathbf{n}-\mathbf{1}t_1)\bigr),\\
&\bigl(\relu(-\mathbf{n}+\mathbf{1}(t_2+\eta))-\relu(-\mathbf{n}+\mathbf{1}t_2)\bigr)\Bigr),\\
\end{aligned}
\end{equation}
where
\begin{equation}
\begin{aligned}
&\relu(x)=\begin{cases}
  x, & \textrm{if } x>0,\\
  0, & \textrm{otherwise}.
\end{cases}
\end{aligned}
\end{equation}
and $\mathbf{n}=\bmat{0,1,...,l}$ is a fixed vector containing $l+1$ consecutive integers from $0$ to $l$, $\mathbf{1}\in\real{l+1}$, $l$ is the length of the signal, $\eta \in\{\real{}-\{0\}\}$ is a hyperparameter tuned for the slope of the time function and can be used to adjust the learning rate of the time function (see Fig~\ref{fig:timefunc} for an illustration).

The output of the time function can be used as an indicator vector in \eqref{stl-ind:sub5} and \eqref{stl-ind:sub6}:
\begin{equation}\label{eq:time}
    \mathbf{w_t}=[w_0,w_1,...,w_{l}]\in \real{l+1}
\end{equation}

\begin{figure}[h]
    \centering
    \includegraphics[scale=0.4]{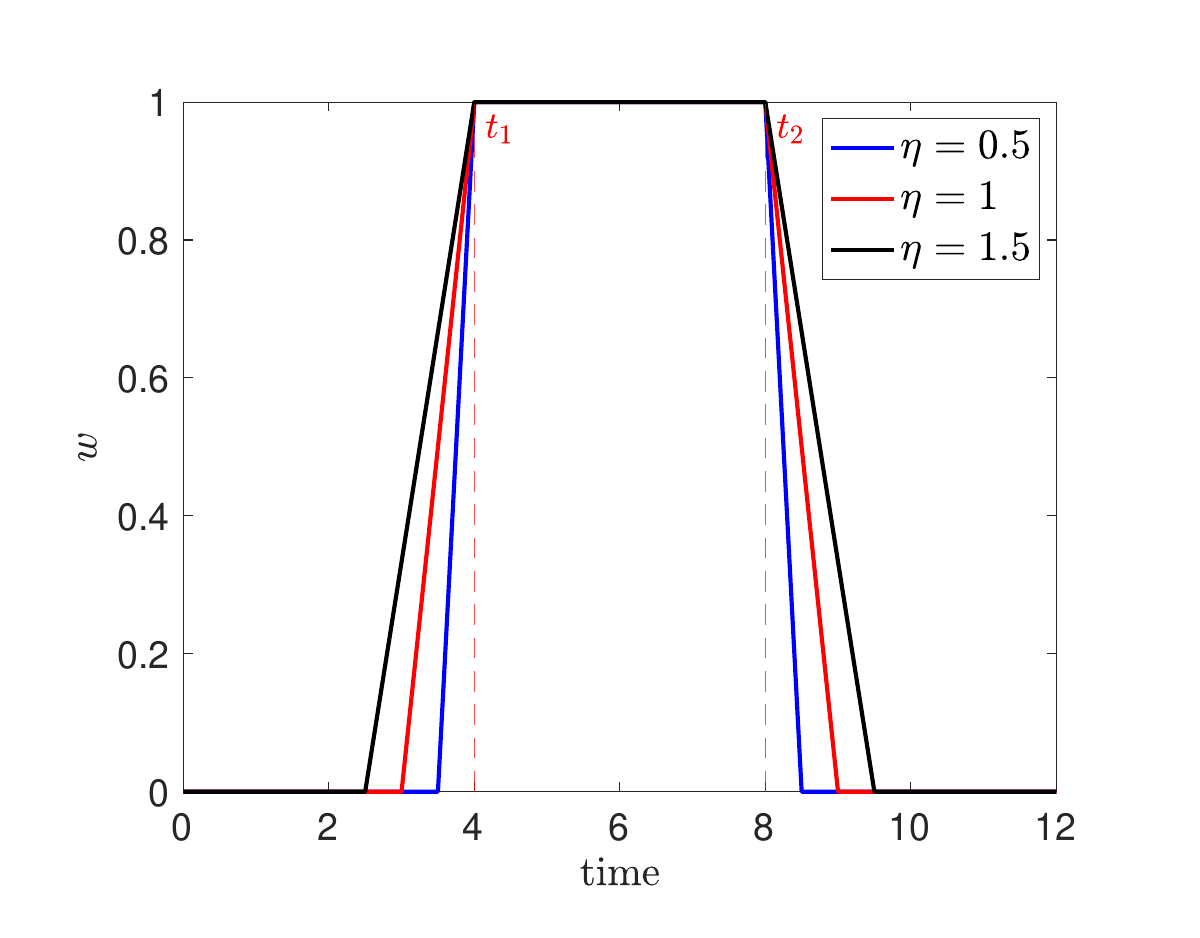}
    \caption{Time function with different $\eta$ for a signal of length $12$ with $t_1=4$ and $t_2=8$.}
    \label{fig:timefunc}
\end{figure}

\subsection{Activation Functions for Logical Operator}\label{logical}
In this section, we discuss in detail the activation functions \eqref{and}, \eqref{or}, the \emph{conjunction-disjunction gate matrix} $M^g$ and the \emph{gate function} $p(\cdot)$.

The conjunction-disjunction gate matrix $M^g$ is a \emph{real-valued} matrix defined as:
\begin{equation}\label{eq:c-d gate matrix}
    M^g = \bmat{\mathbf{c}_1^g\\\mathbf{c}_2^g\\\vdots\\\mathbf{c}_m^g}\\
    = \bmat{c_{11}^g &c_{12}^g &... &c_{1k}^g\\
    c_{21}^g &c_{22}^g &... &c_{2k}^g\\
    \vdots &\vdots &\ddots &\vdots\\
    c_{m1}^g &c_{m2}^g &... &c_{mk}^g\\}\in\real{m\times k}
\end{equation}
where $m$ is the maximum number of terms of the learned STL formula in DNF, $k$ is the maximum number of temporal operators.

We introduce the real-valued matrix $M^g$ to generate the binary conjunction-disjunction matrix $M$ in \eqref{eq:c-d matrix}. To obtain the binary matrix $M$, the method from \cite{danyang:sparsenn} is employed. We introduce $c_{ij}^g$ as variables drawn from a Bernoulli distribution. We consider $c_{ij}$ as gate variables sampled from the Bernoulli random variables $c_{ij}^g$. The Bernoulli function is defined as:
\begin{equation}\label{eq:ber2}
    bernoulli(c_{ij}^g)=
    \begin{cases}
      1 & \textrm{if } 0.5\leq c_{ij}^g\leq 1,\\
      0 & \textrm{if } 0\leq c_{ij}^g< 0.5.
    \end{cases}
  \end{equation}
We learn the real-valued variables $c_{ij}^g$ instead of learning the binary weight variables $c_{ij}$. We estimate the gradient of step $c \sim bernoulli(c^g)$ using the straight-through estimator~\cite{danyang:STE}. To guarantee $c$ falls in $[0,1]$, we pass $c$ through a clip function after performing the optimizer step. The clip function is:
\begin{equation}\label{eq:clip}
\clip(c)=
    \begin{cases}
      1 & \textrm{if } c\geq 1,\\
      0 & \textrm{if } c\leq 0,\\
      c & \textrm{otherwise}.
    \end{cases}
\end{equation}
The gate function $p(c^g)$ is defined as:
\begin{equation}
    p(c^g) = bernoulli(\clip(c^g)),
\end{equation}
which can be used to obtain the indicator vectors $\mathbf{c}_i,i=1,...,m$ in \eqref{eq:c-d matrix}:
\begin{equation}
    \mathbf{c}_i=p(\mathbf{c}_i^g)
\end{equation}
We can learn the logical operators by simply learning the variables of the conjunction-disjunction gate matrix $M^g$. The complexity of the neural network scales linearly with the number of temporal operators, and also linearly with the number of "disjunction" logical operators. The NN-TLI can efficiently learn highly complex STL formulas.

\begin{figure*}[h]
  \centering
  \subfloat["Go Forward" and "Stop and Go"]{\includegraphics[scale=0.25]{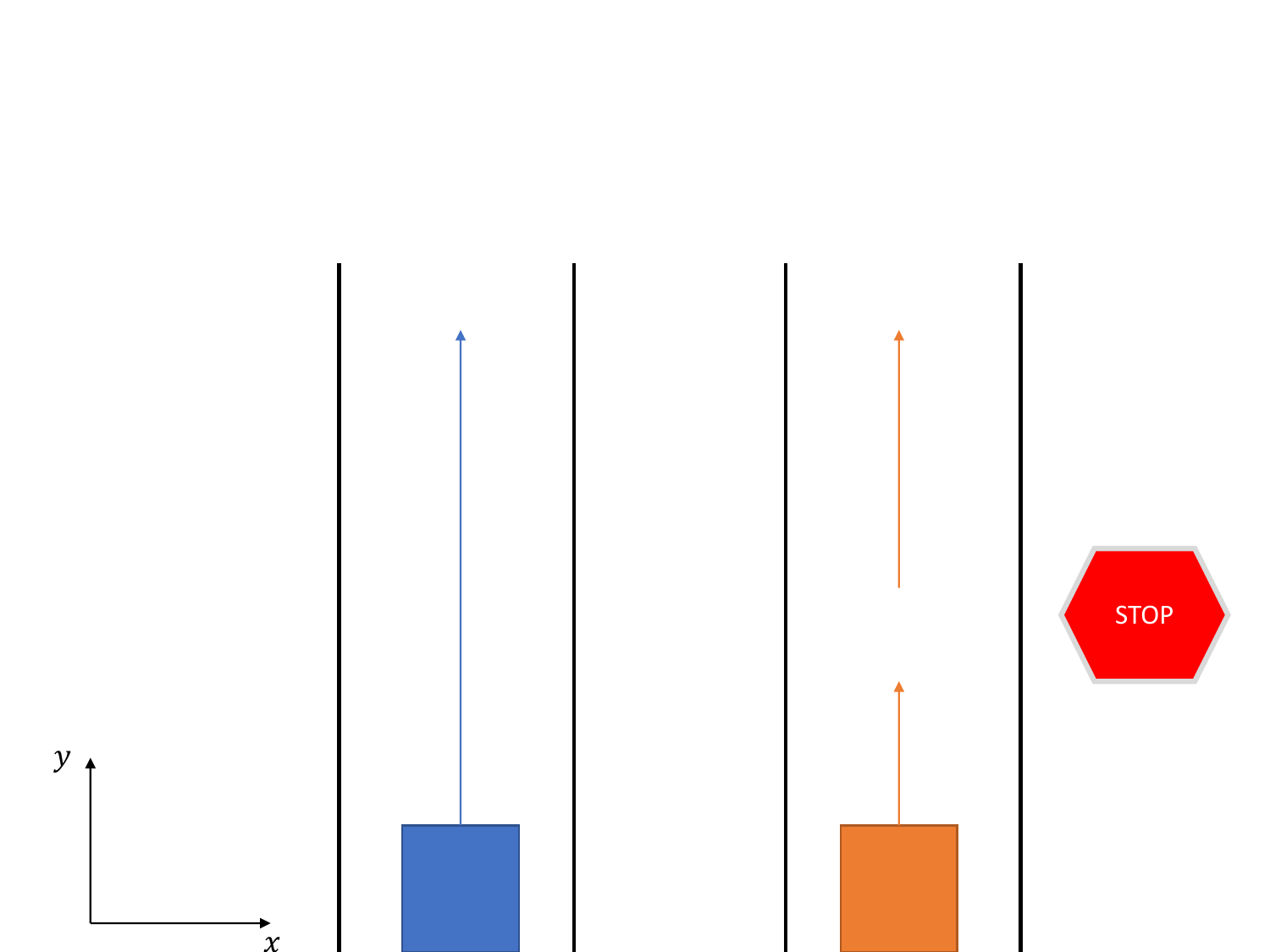}\label{fig:s1}}
  \subfloat["Left turn to lane 1" and "Left Turn to lane 2"]{\includegraphics[scale=0.25]{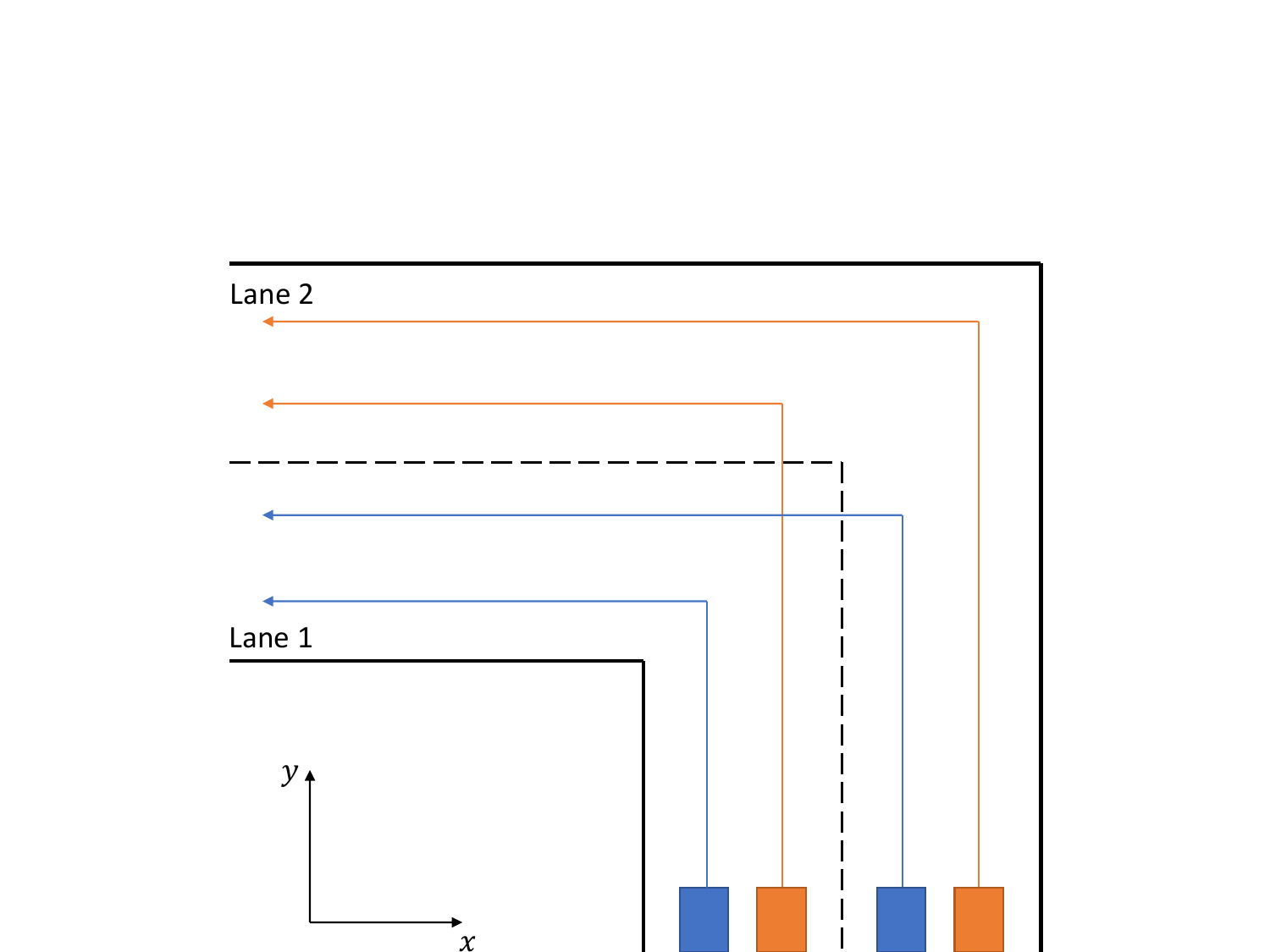}\label{fig:s2}}
  \subfloat["Switch Lane" and "Overtake"]{\includegraphics[scale=0.25]{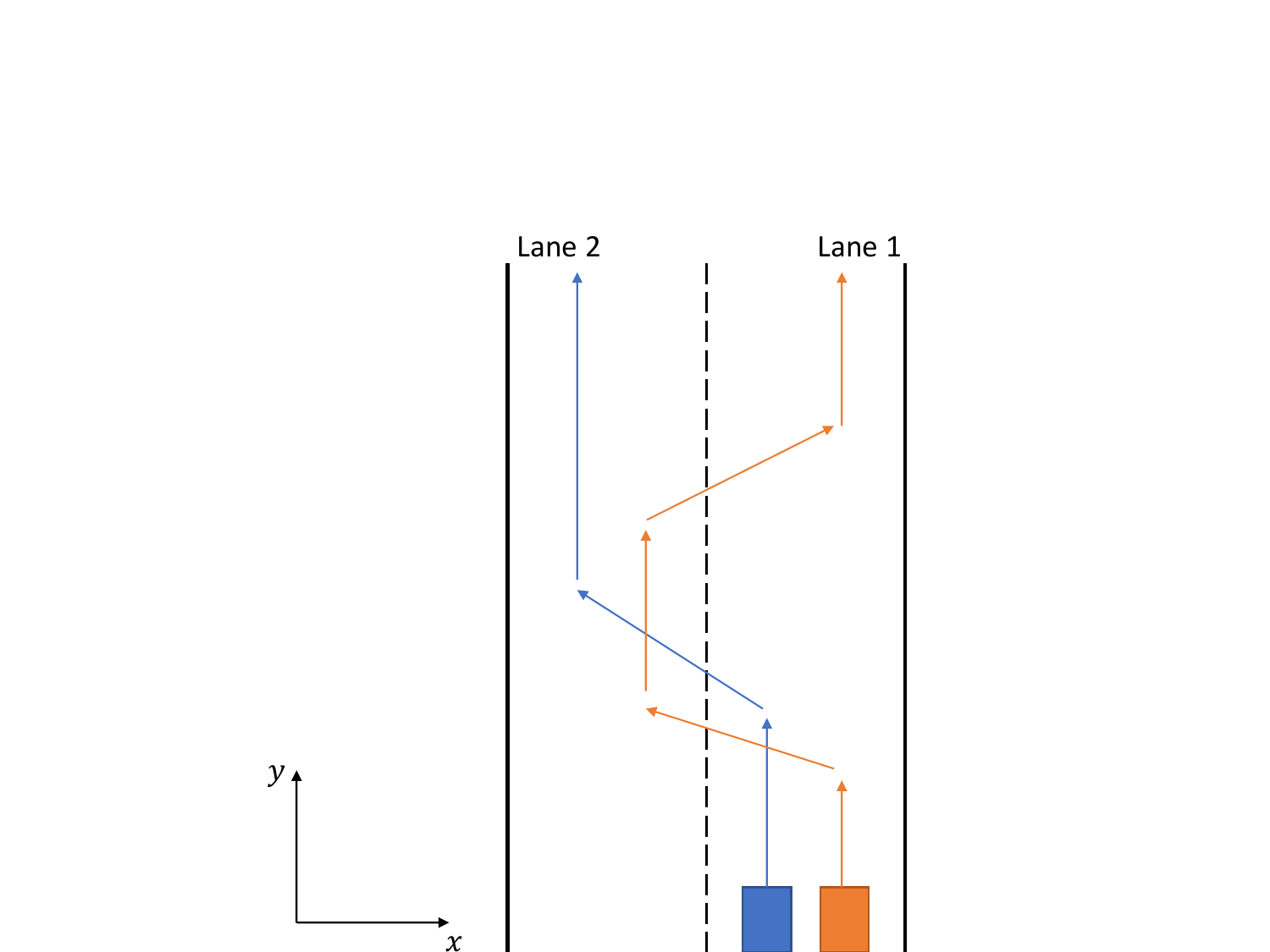}\label{fig:s3}}
  \caption{Examples of driving behaviors. The blue trajectories in (a), (b), (c) represent "Go Forward", "Left turn to lane 1" and "Switch Lane", respectively. The orange trajectories in (a), (b), (c) represent "Stop and Go", "Left Turn to lane 2" and "Overtake", respectively.}
  \label{fig:driving}
\end{figure*}

\section{LEARNING OF NN-TLI}\label{sec-4}
\label{sec:nn-tli-training}

We use Pytorch to build the NN-TLI and perform the back-propagation. The code can be found at \url{https://github.com/danyangl6/NN-TLI}.
\subsection{Loss function}
We label the two sets of data as 1 and -1, respectively. The purpose of the learning process is to minimize the loss function is defined as:
\begin{equation}
    l = e^{-yr_{\chi}},
\end{equation}
where $y$ is the label of the data, $r_{\chi}$ is the robustness degree of the learned STL formula.

\subsection{Formula simplification}

To obtain a compact STL formula, we introduce a method to post-process the conjunction-disjunction matrix $M$ after the training procedure. The formula simplification algorithm is shown in Algorithm~\ref{simplify}. By leveraging this method, we can remove the redundant sub-formulas from the STL formula and extract the meaningful features of the data.
\begin{algorithm}
  \begin{algorithmic}[1]
    \Require The conjunction-disjunction matrix $M$, number of terms $m$, number of temporal operators $k$, the misclassification rate $MCR$ using $M$.
    \Ensure Sparse conjunction-disjunction matrix $\tilde{M}$.
    \State $\tilde{M}\gets M$
    \For{$i=1,2,...,m$}
    \For{$j=1,2,...,k$}
    \State $M'\gets \tilde{M}$
    \State $M_{ij}'\gets 0$
    \State Compute the new misclassification rate $MCR'$\par
    \hskip\algorithmicindent using $M'$.
    \If{$MCR'= MCR$}
    \State $\tilde{M}\gets M'$
    \EndIf
    \EndFor
    \EndFor
  \end{algorithmic}
  \caption{Formula Simplification Algorithm}
  \label{simplify}
\end{algorithm}
\begin{figure}[h]
    \centering
    \includegraphics[scale=0.3]{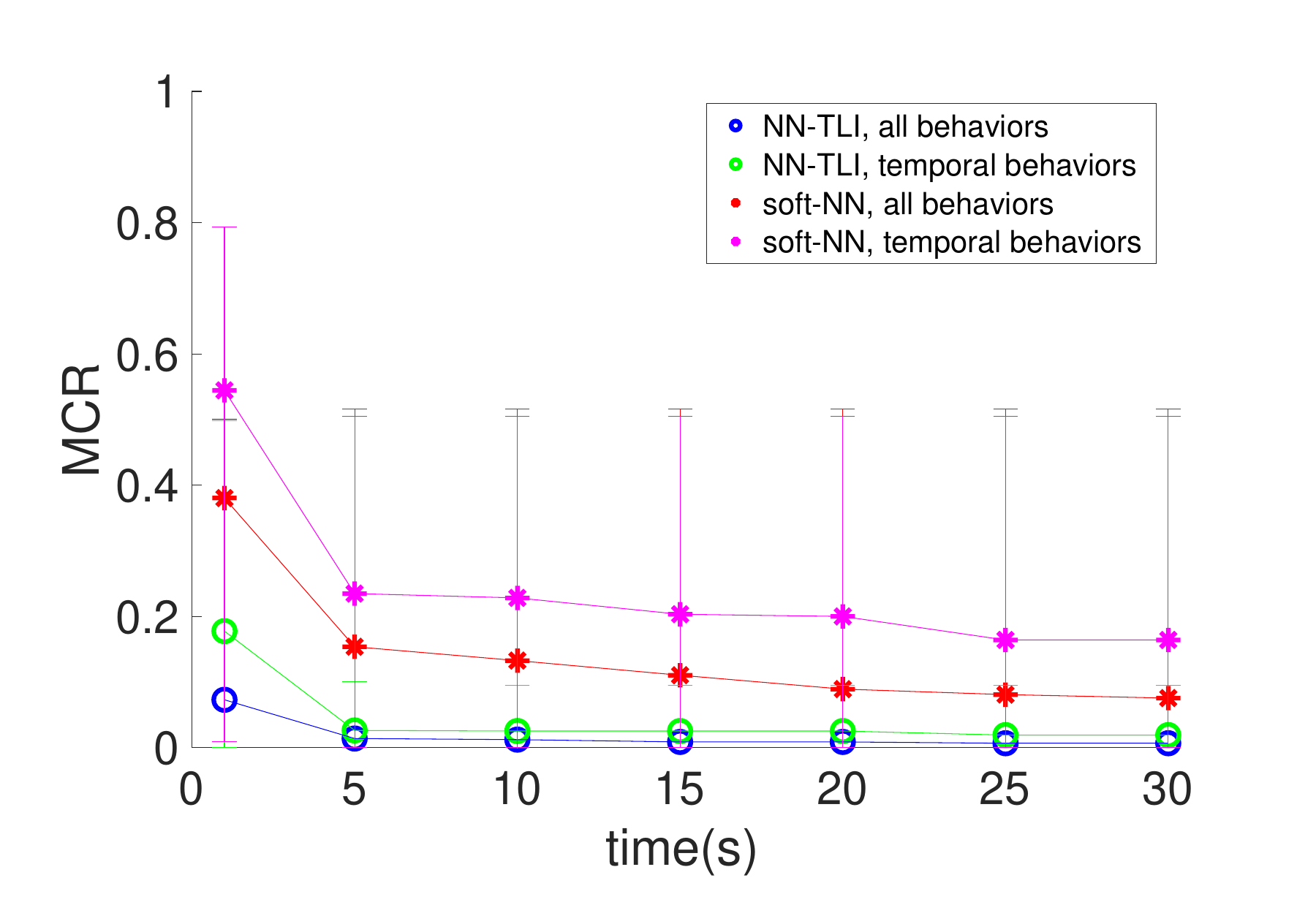}
    \caption{Result comparisons between baselines NN-TLI and soft-NN. Blue and green points represent the mean $MCR$ of NN-TLI. Red and magenta points represent the mean $MCR$ of soft-NN. 
    The error bars represent the highest and lowest $MCR$ for each case.}
    \label{fig:acc1}
\end{figure}
\section{Case study}\label{sec-5}
In this section, we first introduce a synthetic dataset containing general behaviors under autonomous driving scenarios. We use our proposed NN-TLI to classify these behaviors and interpret the characteristic from the learned STL formula. We compare our method with \cite{yan2021neural} to show the benefits of our sparse softmax function. We also experiment on a dataset of naval surveillance scenarios from \cite{kong2016temporal} and compare the results with existing methods~\cite{kong2016temporal,aasi2021classification} in the sense of $MCR$, efficiency, and interpretability.

\subsection{Autonomous driving scanarios}
We include six kinds of driving behaviors in three scenarios in the driving scenario's dataset. The signals are 2-dimensional trajectories shown in Fig \ref{fig:driving}. In scenario 1, "Go Forward" vehicle keep going straight; "Stop and Go" vehicle stop at the stop sign for three time points and then go straight. In scenario 2, the vehicles first go forward, then make a left turn and enter either lane 1 or lane 2.  In scenario 3, "Switch Lane" vehicles first go forward in lane 1, then switch to lane 2 and end up driving in lane 2, "Overtake" vehicles initially go forward in lane 1, then enter lane 2, drive back to lane 1 and end up driving in lane 1. The autonomous vehicle is considered to start at a random initial position with random velocities. For each kind of behavior, we generate 2000 samples of length 40. The code for generating the dataset can be found here \url{https://github.com/danyangl6/NN-TLI}.

We use the proposed NN-TLI to classify between two combinations of these six behaviors, that is 30 test cases in total. We compare our NN-TLI with the network using activation functions based on the traditional softmax function. We call it soft-NN for short reference. The parameter $\beta$ in activation functions is the same for both methods. We evaluate these methods based on $MCR$ and computation time. The results are shown in Fig \ref{fig:acc1}. The blue and red points with error bar represent $MCR$ of NN-TLI and soft-NN, respectively, which indicates that compared to soft-NN, the proposed NN-TLI can reach a lower mean $MCR$ in a shorter time, thanks to its soundness. The proposed NN-TLI is relatively stable, with 10\% difference between the maximum and minimum $MCR$, as opposed to the 50\% of soft-NN.

It is worth noting that temporal property is not necessary when classifying some of the behaviors. For example, we can classify "Go forward" vs "Turn left to lane 1" with a fixed time interval $[0,40]$. To address this, we intentionally pick the results between behaviors requiring temporal characteristics to classify, e.g. "Go Forward" vs "Stop and Go". The results are also shown in Fig \ref{fig:acc1}. The green and magenta points with error bar represent $MCR$ of NN-TLI and soft-NN, respectively. The results show that the mean $MCR$ of soft-NN dramatically increases when classifying these behaviors since the traditional softmax function is not able to properly approximate the maximum within the specified time interval.

We can straightforwardly interpret the property of the data from the STL formula learned by NN-TLI. For example, the learned formula to classify "Go forward" and "Overtake" with lane 1 between $[-2,2]$ and lane 2 between $[-6,-2]$ is
\begin{equation}
    \always_{[0,39]} (x>-1.97)
\end{equation}
We can describe the property of "Go forward" as "the vehicle always stays at lane 1 for all time."
\begin{figure}[h]
    \centering
    \includegraphics[scale=0.3]{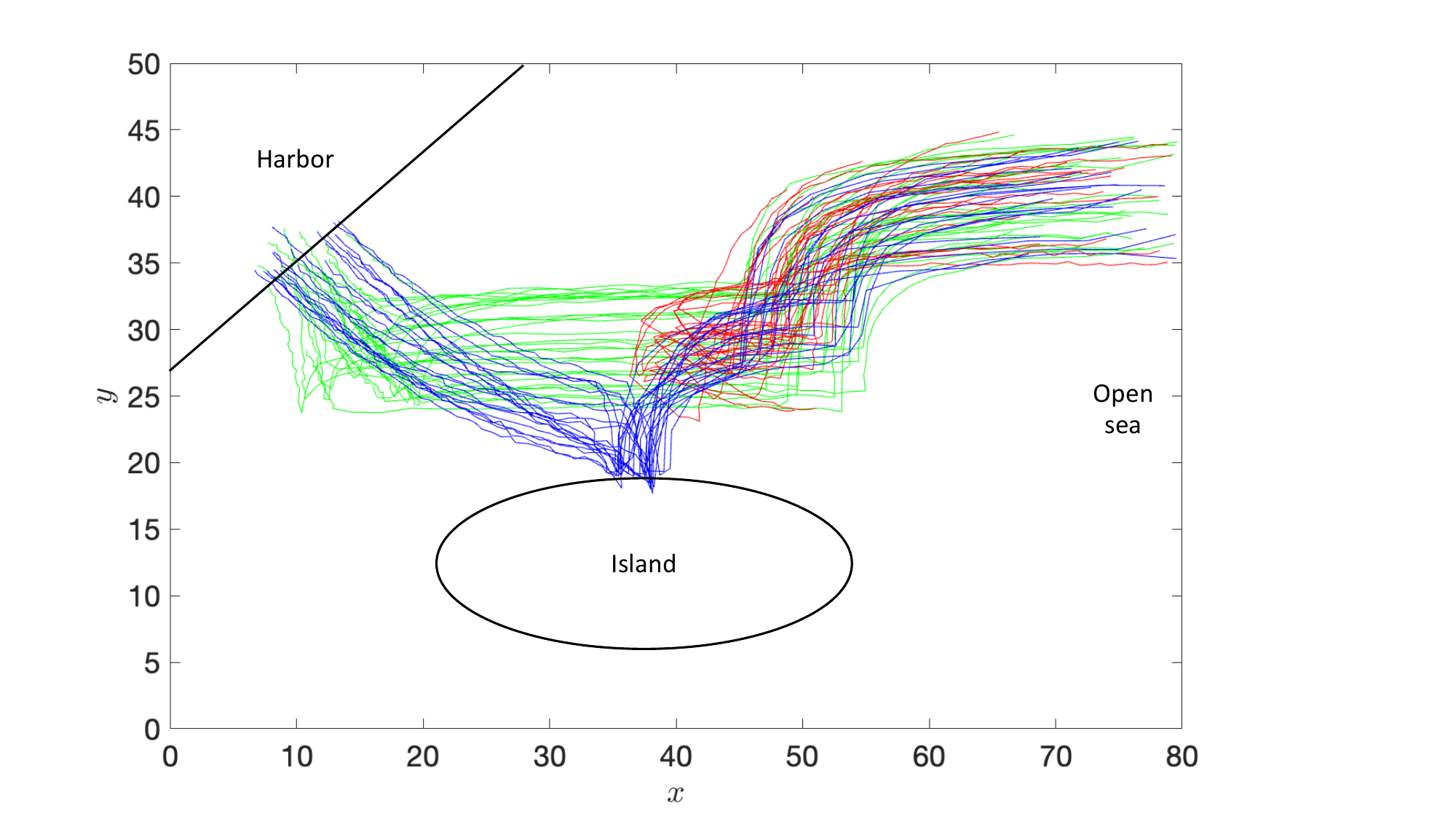}
    \caption{Naval surveillance scenario}
    \label{fig:naval}
\end{figure}
\begin{table*}[h]
\caption{Comparison results between NN-TLI and other interpretable classification methods}
\label{table:navalstl}
\centering
{
\hspace*{-0.2cm}
\begin{tabular}{rccc}
  \toprule
  Method &$MCR$ &Time(min) &STL formula\\
  \midrule
  NN-TLI &0.0000 &0.2 &$\always_{[9,14]}(y>23.37) \land \event_{[60,60]}(x<27.96)$\\
  soft-NN &0.0075 &1.43 &$\always_{[0,15]}(y>23.72) \land \event_{[60,60]}(x<20.23)$\\
  BCDT &0.0100 &33 &$\event_{[28,53]}(x\leq 30.85)\land \always_{[2,26]}((y>21.31) \land (x>11.10))$\\
  On-line learning &0.0885 &16 &$\event_{[0,38]}(\always_{[21,27]}(y>19.88) \land \always_{[11,36]}(x<34.08))$\\
  \bottomrule
  \end{tabular}
  }
\end{table*}

\subsection{Naval surveillance scenario}
To compare with the existing interpretable classification methods~\cite{kong2016temporal,aasi2021classification}, we conduct experiments on the naval surveillance scenario dataset from their paper. The naval dataset is consisting of 2-dimensional trajectories representing various vessel behaviors. As shown in Fig \ref{fig:naval}, the green trajectories are normal behaviors, labeled as "1", representing vessels from the open sea and heading directly to the harbor. The blue and red trajectories are abnormal trajectories, labeled as "-1". The blue trajectories show vessels veering to the island and the red trajectories show vessels returning to the open sea after approaching island. We use our proposed NN-TLI to classify normal and abnormal behaviors and compare with the results from boosted concise decision tree method(BCDT)~\cite{aasi2021classification} and online learning method~\cite{kong2016temporal}.

We provide the learned STL formula from our NN-TLI and methods from \cite{kong2016temporal,aasi2021classification} in Table \ref{table:navalstl}. The learned STL formula from our NN-TLI is
\begin{equation}\label{navalstl}
    \always_{[9,14]}(y>23.37) \land \event_{[60,60]}(x<27.96)
\end{equation}
We can explain the characteristic of normal vessel behaviors with human language using \eqref{navalstl}: The vessels eventually reach the harbor in the end and always do not reach the island between $9$ seconds and $14$ seconds. The learned STL formula of NN-TLI is solid while compact. We compare the time to reach the lowest $MCR$ during the learning procedure, as shown in Table \ref{table:navalstl}. The proposed NN-TLI can reach 0\% $MCR$ at $12$ seconds, which is $150\times$ faster than boosted concise decision tree method in \cite{aasi2021classification} and $80\times$ faster than on-line learning algorithm in \cite{kong2016temporal}, which only reach 8.85\% $MCR$.

\section{CONCLUSIONS AND FUTURE WORK}\label{sec-6}

In this paper, we present an interpretable neural network to classify time-series data, which can be directly translated into an STL formula. By leveraging the novel activation functions, we can easily learn the STL formula from scratch using off-the-shelf machine learning tools while guaranteeing soundness. We demonstrate that our proposed method can interpret the properties of data with high computation efficiency and accuracy on a synthesis dataset based on autonomous driving scenarios and a naval scenarios dataset. In the future, we would like to extend this framework to represent more complicated properties of time-series data.



\bibliographystyle{IEEEtran}
\bibliography{biblio/NNTLI}

\end{document}